\documentclass[conference]{IEEEtran}
\usepackage{lipsum}
\usepackage{graphicx,epstopdf,tikz,pgfplots,subfig,amssymb,amsfonts,amsmath,mathrsfs,lettrine}
\usepackage[makeroom]{cancel}
\usepackage[T1]{fontenc}
\usepackage[utf8]{inputenc}
\usepackage{authblk}
\usepackage[font=small,labelfont=bf]{caption}
\usepackage{float}
\usepackage{cite}
\usepackage{url}

\ifCLASSINFOpdf

\else

\fi
\newtheorem{theorem}{Theorem}
\newtheorem{lem}{Lemma}

\newtheorem{col}{Corollary}
\newcommand{\ignore}[1]{}

\begin{document}
%
\title{ Handover Management in Dense Cellular Networks: A Stochastic Geometry Approach}

\author[Rabe Arshad]
       {Rabe Arshad$^*$, Hesham ElSawy$^{\dag}$, Sameh Sorour$^*$,
       Tareq Y. Al-Naffouri$^{\dag}$, and Mohamed-Slim Alouini$^{\dag}$
       \\ \\
       $^*$Electrical Engineering Department, King Fahd University of Petroleum and Minerals (KFUPM), Saudi Arabia\\
       Emails: \{g201408420, samehsorour\}@kfupm.edu.sa\\
       $^{\dag}$CEMSE Division, EE program, King Abdullah University of Science and Technology (KAUST), Saudi Arabia\\
       Emails: $\{$hesham.elsawy, tareq.alnaffouri, slim.alouini$\}$@kaust.edu.sa}

\maketitle

\begin{abstract}
Cellular operators are continuously densifying their networks to cope with the ever-increasing capacity demand. Furthermore, an extreme densification phase for cellular networks is foreseen to fulfill the ambitious fifth generation (5G) performance requirements. Network densification improves spectrum utilization and network capacity by shrinking base stations' (BSs) footprints and reusing the same spectrum more frequently over the spatial domain. However, network densification also increases the handover (HO) rate, which may diminish the capacity gains for mobile users due to HO delays. In highly dense 5G cellular networks, HO delays may neutralize or even negate the gains offered by network densification. In this paper, we present an analytical paradigm, based on stochastic geometry, to quantify the effect of HO delay on the average user rate in cellular networks. To this end, we propose a flexible handover scheme to reduce HO delay in case of highly dense cellular networks. This scheme allows skipping the HO procedure with some BSs along users' trajectories. The performance evaluation and testing of this scheme for only single HO skipping shows considerable gains in many practical scenarios.\\
\end{abstract}

\IEEEpeerreviewmaketitle

\section{Introduction}
\lettrine{N}{       etwork} densification via base stations (BSs) deployment has always been a viable solution for cellular operators to cope with the increasing capacity demand. It is also expected that cellular operators will rely on network densification to fulfill a big chunk of the ambitious 5G requirements~\cite{1a},\cite{2a}. Network densification can be achieved by deploying different types of BSs (e.g., macro, micro, pico, and femto) according to the performance, time, and cost tradeoffs \cite{3a}. Deploying more BSs decreases the service area (also referred to as the footprint) of each BS, which improves the spatial spectral utilization and network capacity. However, the network densification gains come at the expense of increasing the handover rate for mobile users.\\

Handover (HO) is the procedure of changing the association of mobile users from one BS to another to maintain their connection to their best serving BS. According to the network objective (e.g., best signal strength, equal BSs' load, low service delay, etc), different association strategies may be defined to determine the best serving BS \cite{4a},\cite{5a}. In all cases, HO is employed to update users' association with mobility to satisfy the defined network objective along their path. The HO procedure involves signaling overhead between the mobile user, the serving BS, the target BS, and the core network elements, which consumes resources and introduces delay. Note that the HO rate increases with the intensity of BSs, which imposes higher delays and negatively affects the mobile user's throughput. In the case of highly dense cellular networks, the HO delay may become a performance limiting parameter that may neutralize or even negate the gains offered by network densification. Therefore, the HO delay should be carefully incorporated into the design of dense cellular networks to visualize and mitigate its effect. In particular, novel HO strategies are required to reduce the HO delay in order to harvest the foreseen network densification gains.\\

In order to draw rigorous conclusions on the HO effect and design efficient HO schemes, mathematical paradigms that incorporate the effect of HO delay into the network key performance indicators are required. In this regards, recent advances in modeling cellular networks using stochastic geometry can be exploited to develop such mathematical frameworks. Stochastic geometry has succeeded in providing a systematic analytical paradigm to model and design cellular networks \cite{5a,7a,8a,9a,10a}, see \cite{6a} for a survey. In stochastic geometry analysis, the network is abstracted to a convenient point process that maintains a balance between practicality and tractability. The Poisson point process (PPP) is well understood and widely accepted due to its tractability and simplicity. Stochastic geometry with BSs modeled as a PPP has enriched the literature with valuable results that enhanced our understanding of the cellular network behavior. For instance, the PPP assumption is reinforced by several experimental studies \cite{7a,8a,9a}. Coverage probability and ergodic capacity are studied in \cite{7a, 5a} for downlink cellular networks and in \cite{uplink_H} for uplink cellular networks. Symbol error probability and bit error probability for cellular networks are investigated in \cite{eid_app, Laila_Uplink, Laila_letter}. 
Cellular networks with D2D communication is characterized in \cite{D2D_h, Jeffery_D2D}. 
However, the effect of HO delay on the performance of dense cellular networks has been overlooked.\\

\subsection{Contribution}

In this paper, we study the effect of HO delay on the user average rate in dense cellular network environments. In particular, we develop a mathematical paradigm, based on stochastic geometry, that incorporates the HO delay into the rate analysis. It is worth noting that the signal-to-interference-plus-noise-ratio (SINR) dependent expressions are derived based on the stationary PPP analysis. However, we show by simulation that the stationary expressions capture the SINR performance of mobile users almost exactly. This implies that averaging over all users' trajectories in all network realizations is equivalent to averaging over all users locations in all network realizations. To this end, the results manifest the prominent negative effect of HO delay on the rate at high intensities of BSs.

Therefore, we propose a flexible HO scheme denoted as {\em HO skipping}, which allows users to skip associating with some of the BSs along their trajectories to reduce the HO delays. That is, a user can sacrifice being always best connected to decrease the HO delay and improve the long-term average rate. In this paper, we will consider a single HO skipping scheme in which a user, connected to its best serving BS, skips associating to only one subsequent best serving BS on its trajectory and must then reconnect to the following one. This pattern is repeated along the user's trajectory such that it alternates between connecting to its best serving BS and skipping the following one (cf. Fig.~\ref{1d} and Fig.~\ref{voron}).

For both the conventional HO and single HO skipping schemes, we derive expressions for the coverage probability and average throughput. The results show that, although BS skipping reduces the average coverage probability of mobile users, it improves their long-term average rate at high speeds and/or high network densities. The developed model can be used to decide the threshold at which BS skipping is beneficial and quantify the associated performance gains.

It is important to note that this work focuses on a single tier cellular network and a single HO skipping scheme. The case of multi-tier cellular networks and multiple HO skipping rules will be the subject of future work.
\begin{figure}
\centering
\includegraphics[width=0.9 \linewidth]{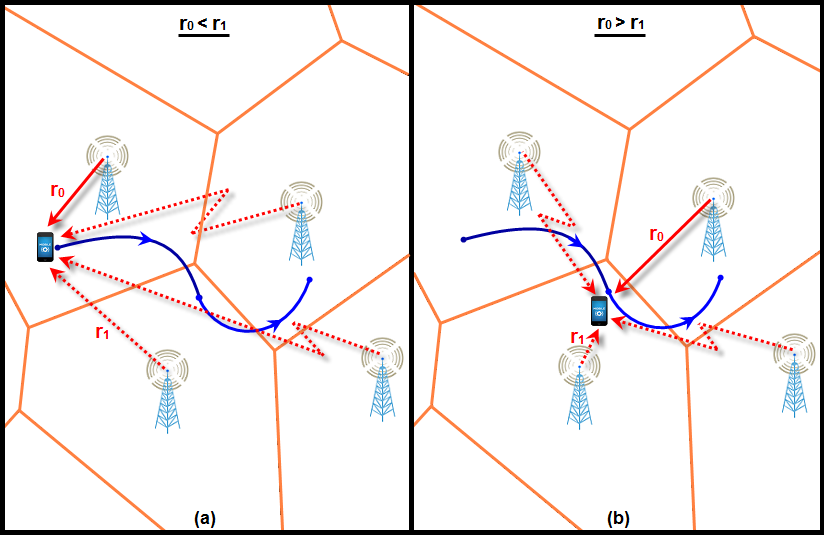}
\small \caption{(a) represents best connected case. (b) shows the blackout case. In both cases, $r_0$ represents the distance between test user and its serving BS, $r_1$ represents the distance between test user and its first interfering BS and blue solid line represents the user trajectory.}
\label{1d}
\end{figure}



\section{System Model}
In this article, we consider a single-tier downlink cellular network in which the BSs' locations are abstracted via a two dimensional PPP $\Phi$ of intensity $\lambda$. All BSs are assumed to transmit with the same effective isotropic radiated power (EIRP) denoted by $P$. A general power-law path loss model with path loss exponent $\eta > 2$ is considered. In addition to the path-loss, transmitted signals experience multi-path fading. We assume a Rayleigh fading environment such that the channel power gains have independent exponential distributions. The users associate according to the well-known radio signal strength (RSS) association rules. In the depicted single-tier case, the best RSS association boils down to the nearest BS association strategy. The association regions for each BS can be visualized via the voronoi tessellation diagram shown in Fig. 2.

The analysis is conducted on a test mobile user\footnote{By Slivnyak’s theorem for PPP \cite{7a,19a}, there is no loss of generality to conduct the analysis for a test user} which is moving on an arbitrary trajectory with velocity $v$. We ignore HO failures due to blocking and focus on the coverage probability (defined as the probability that the SINR exceeds a certain threshold $T$) and the ergodic rate (defined by Shannon's capacity formula). That is, we assume that all BSs can assign a channel to the test user when this user passes within its coverage range. Furthermore, when the test user is assigned to a certain channel from a generic BS, we assume that all other BSs across the spatial domain are reusing the same channel.

Due to the HO procedure, a delay time $d$ is imposed on the test user, at which no useful data is transmitted to it. The delay $d$ represents the time spent in HO signaling between the serving BS, target BS, and core network elements. It is important to note that the delay $d$ may depend on the type of BSs. For instance, the HO delay is small in macro BSs, which are microwave or optically backhauled to the core network. In contrast, the delay $d$ may be more significant in the case of internet protocol (IP) backhauled small BSs (i.e., pico or femto) \cite{IP_backhaul}.


In the {\em conventional} HO case, the test user maintains the best connectivity throughout its trajectory. Hence, when the test user passes by a certain voronoi cell, it connects to the BS in the center of that voronoi cell. This implies that $P r_{0}^{-\eta} > P r_{i}^{-\eta}$, $\forall i \neq 0$ is always satisfied, where $r_{0}$ and $r_{i}$ denote the distances from the test user to the serving and $i^{th}$ interfering BSs, respectively. We propose the {\em skipping} HO scheme, where the test user skip associating to one BS after every HO execution. This implies that $P r_{0}^{-\eta}$ > $P r_{i}^{-\eta}$, $\forall i \neq 0$ is satisfied for $50 \%$ of the time and $P r_{1}^{-\eta} > P r_{0}^{-\eta} > P r_{i}^{-\eta}$, $\forall i \notin \{0,1\}$ is satisfied for the rest of the time, where $r_1$ denotes the distance from the test user to the skipped BS when it passes within the voronoi cell of the skipped BS. Thus, in the skipping mode, the user alternates between the blackout and best connected state along its trajectory. Here, we denote the event at which the test user is located within the voronoi cell of a skipped BS as {\em blackout} event. Fig.~\ref{1d} shows the best connected and blackout events. Fig.~\ref{voron} depicts the conventional (green dotted line) and skipping (blue line) HO schemes for the red trajectory.

\begin{figure}[!t]
\centering
\includegraphics[width=0.85 \linewidth]{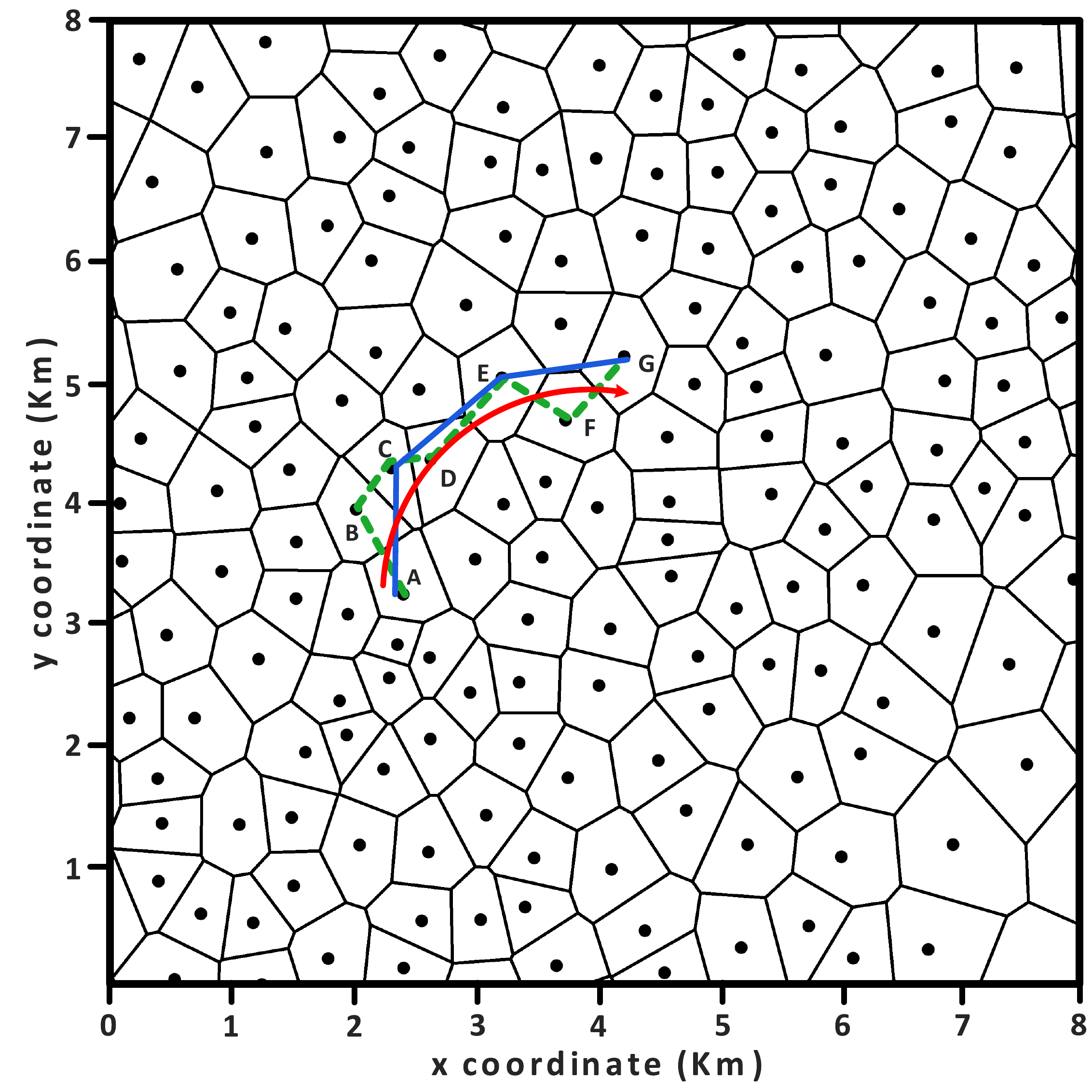}
\small \caption{Voronoi tessellation of actual cellular network (Macros only) in urban area with $\lambda$=3 BS per km\textsuperscript{2}. Red solid line depicts user trajectory while green dotted and blue solid lines show conventional and skipping HOs respectively. In the conventional (skipping) scheme the test user associates to the BSs $\{A, B, C, D, E, F, G\}$ $\left(\{A, C, E, G\}\right)$. }
\label{voron}
\end{figure}

\section{Service Distance Distribution}

The first step in the analysis is to characterize the service distance in the best connected and blackout cases. Note that the service distance is random due to the irregular network topology along with user mobility. It is important to characterize the service distance as it highly affects the SINR. Particularly, when the user is best connected, the RSS association creates an interference protection of radius $r_0$ around the user. In the blackout case, the user keeps its association with the same BS when it enters the voronoi cell of another BS. Hence, the skipped BS is closer to the user than its serving BS (i.e., $r_1$ < $r_0$). In other words, the skipped BS is located within the radius $r_0$ and every other interfering BS is located outside $r_0$. Note that, in the blackout case, $r_1$ and $r_0$ are correlated since $r_1 < r_0$. The distribution of serving BS in the best connected case and the joint distribution of the distances from the test user to the skipped and serving BS in the blackout case are characterized via the following lemma:

\begin{lem}
\label{lem:joint}
In a single tier cellular network with intensity $\lambda$, the distance distribution between a best connected user and its serving BS is given by:
\begin{align} \label{ser1}
 f^{(c)}_{r_{0}}( r ) &= 2\lambda \pi r e^{-\lambda\pi r^2},\quad 0 \leq r \leq \infty
\end{align}
and the joint distance distribution between a user in blackout and its serving and skipped BSs is given by:
\begin{align}
f^{(bk)}_{r_0,r_1}( x,y)
=  {4 (\pi \lambda)^{2}  x y e^{-\pi \lambda x^2}}; \quad 0 \leq y \leq x \leq \infty
\label{eq:joint}
\end{align}
\end{lem}

\begin{proof}
The PDF $f^{(c)}_{r_{0}}(. )$ is obtained from the null probability of PPP as in \cite{7a}. The joint PDF $f^{(bk)}_{r_0,r_1}( .,.)$ is obtained by writing the conditional CDF of $r_0$ given $r_1$ as $ \mathbb{P}\left\{ r_0 < x | r_1\right\}
=  {1 - e^{-\pi \lambda (x^2 -r_1^2)}}$. Then, differentiating the conditional CDF of $r_0$ with respect to $x$ and multiplying by the marginal PDF\footnote{Since the user is in blackout, $r_1$ is the distance to the closest BS and $r_0$ is the distance to the second nearest BS. Hence, the marginal PDF of $r_1$ is given by $f_{r_{1}}(y) = 2\lambda \pi y e^{-\lambda\pi y^2}$ } of $r_1$, the joint PDF in \eqref{eq:joint} is obtained.
\end{proof}
The marginal and conditional service distance distributions for the blackout case are characterized by the following corollary:
\begin{col} \label{col_dist}
The marginal PDF of the distance between the test user and its serving BS in the blackout case is given by:
\begin{align}
\hspace{0.3cm}
f^{(bk)}_{r_{0}}(r) =2(\lambda\pi)^2r^3e^{-\lambda\pi r^2},\quad 0 \leq r \leq\infty
\label{ser_dist}
\end{align}
where $r_0$ represents the distance between the test user and its serving BS, which is the second nearest BS in the blackout case.\\
The conditional (i.e., conditioning on $r_0$) PDF of the distance between the test user and the skipped BS in the blackout case is given by:
\begin{align}
\hspace{0.3cm}
f^{(bk)}_{r_{1}}( r|r_0 )=\frac{2r}{r_{0}^{2}},\quad 0 \leq r \leq r_0 \leq \infty
\end{align}
\end{col}
\begin{proof}
The marginal PDF of $r_0$ is obtained by integrating \eqref{eq:joint} with respect to (w.r.t.) $y$ from $0$ to $r_0$ while the conditional PDF of $r_1$ is obtained by dividing \eqref{eq:joint} by the marginal distribution of $r_0$, which is given in \eqref{ser_dist}.
\end{proof}

\section{Coverage Probability}

The coverage probability is defined as the probability that the test user can achieve a specified SINR threshold $T$. For the best connected case, the coverage probability is given by

\begin{eqnarray*}
\hspace{0.4cm} \mathcal{C}_c &=& \mathbb{P} \left\{ \frac{P h_0 r_{0}^{-\eta}}{ \sum_{i\epsilon \Phi \backslash b_0}{}P h_{i}r_{i}^{-\eta} + \sigma^2} \geq T \right\}
\end{eqnarray*}
where $b_0$ is the serving BS. In the blackout case, the coverage probability is given by
\begin{eqnarray*}
\hspace{0.4cm} \mathcal{C}_{bk} &=& \mathbb{P} \left\{ \frac{P h_0 r_{0}^{-\eta}}{P h_1 r_{1}^{-\eta} + \sum_{i\epsilon \Phi \backslash b_0 \text{,}b_1}{}P h_{i}r_{i}^{-\eta} + \sigma^2} \geq T  \right\}
\end{eqnarray*}
where $b_0$ and $b_1$ are the serving and the skipped BSs, respectively. Following \cite{6a}, conditioning on $r_0$ and exploiting the exponential distribution of $h_0$, the conditional coverage probability in the best connected case can be represented as:
\begin{eqnarray} \label{cov1}
\mathcal{C}_c(r_0) = e^{- \frac{ T \sigma^2 r_0^\eta}{P}}\mathscr{L}_{I_r}\left(\frac{T r_0^\eta}{P}\right)
\end{eqnarray}
where $I_r$ is the interference from BSs located outside $r_0$. Similarly, the conditional coverage probability in the blackout case can be represented as:
\begin{eqnarray}  \label{cov2}
\mathcal{C}_{bk}(r_0)= e^{- \frac{ T \sigma^2 r_0^\eta}{P}} \mathscr{L}_{I_1}\left(\frac{T r_0^\eta}{P}\right) \mathscr{L}_{I_r}\left(\frac{T r_0^\eta}{P}\right)
\end{eqnarray}
where $I_1$ is the interference from the skipped BS and $I_r$ is the interference from BSs located outside $r_0$. The conditional (i.e., conditioning on $r_0$) Laplace transforms (LTs) of $I_r$ and $I_1$ are characterized via the following Lemma
\begin{lem} \label{lem:LT}
The Laplace transform of $I_r$ for the best connected and blackout cases is given by:
\begin{eqnarray} \label{IR1}
\mathscr{L}_{I_r}\left(\frac{Tr_{0}^{\eta}}{P}\right) &=& \exp(- \pi\lambda r_{0}^{2}\vartheta(T\text{,}\eta))
\end{eqnarray}
where
\begin{eqnarray}
\vartheta(T\text{,}\eta) &=& T^{2/\eta}\int_{T^{-2/\eta}}^{\infty}\frac{1}{1+w^{\eta/2}}dw\notag
\end{eqnarray}
The Laplace transform of $I_1$ in the blackout case is given by:
\begin{eqnarray} \label{I11}
\mathscr{L}_{I_1}\left(\frac{Tr_{0}^{\eta}}{P}\right)&=& \int_{0}^{r_0}\frac{1}{1+Tr_{0}^{\eta }r^{-\eta}}\frac{2r}{r_{0}^{2}} dr
\end{eqnarray}
\end{lem}
\begin{proof}
See Appendix A.
\end{proof}

In the special case when $\eta=4$, which is a common path-loss exponent for outdoor environments, the LTs in Lemma~\ref{lem:LT} can be represented in a closed form as shown in the following corollary.
\begin{col} \label{int_col}
For the special case of $\eta=4$, the LT in \eqref{IR1} reduces to
\begin{eqnarray}
\mathscr{L}_{I_r}\left(\frac{Tr_{0}^{\eta}}{P}\right)\big|_{\eta=4}= \exp\left(- \pi\lambda r_{0}^{2} \sqrt{T}\arctan\left({\sqrt{T}}\right) \right)
\end{eqnarray}
and the LT in \eqref{I11} reduces to
\begin{align}
\mathscr{L}_{I_1}\left(\frac{Tr_{0}^{\eta}}{P}\right)|_{\eta=4} = 1-\sqrt{T} \arctan\left(\frac{1}{\sqrt{T}}\right)
\end{align}
\end{col}
Using \eqref{cov1}, \eqref{cov2}, and Lemma~\eqref{lem:LT} the following theorem for coverage probability is obtained.

\begin{theorem}
\label{theorem:Coverage Probability}
\begin{figure*}
\begin{eqnarray}
\mathcal{C}_c = \int_0^\infty 2 \pi \lambda x \exp\left(- \frac{ T \sigma^2 x^\eta}{P} - \pi\lambda x^{2} \left(1+ T^{2/\eta}\int_{T^{-2/\eta}}^{\infty}\frac{1}{1+w^{\eta/2}}dw\right) \right) dx
\label{f_cov1}
\\ \notag
\end{eqnarray}

\begin{eqnarray}
\mathcal{C}_{bk} = 4(\lambda\pi)^{2}\int_{0}^{\infty}y e^{-Ty^{\eta}\sigma^{2}-\lambda\pi y^2 (\vartheta(T\text{,}\eta)+1)}\bigg(\int_{0}^{y}\frac{x}{1+Ty^{\eta}x^{-\eta}}dx\bigg)dy
\label{f_cov2}
\\ \notag
\end{eqnarray}
\hrulefill
\end{figure*}
Considering a PPP cellular network with BS intensity $\lambda$ in a Rayleigh fading environment, the coverage probability for the best connected and blackout users can be expressed by \eqref{f_cov1} and \eqref{f_cov2} respectively. For $\eta=4$,
\eqref{f_cov1} and \eqref{f_cov2} reduce to
\begin{eqnarray}
\mathcal{C}_c\big|_{\eta=4}&=& \frac{1}{1+\sqrt{T}\arctan({\sqrt{T}})}\\ \notag \\
\mathcal{C}_{bk}\big|_{\eta=4}&=& \frac{1-\sqrt{T}\arctan(\frac{1}{\sqrt{T}})}{\left(1+ \sqrt{T}\arctan\left(\sqrt{T}\right) \right)^2}\end{eqnarray}
\end{theorem}
\begin{proof}
We prove the theorem by substituting the LTs from Lemma~\ref{lem:LT} and Corollary~\ref{int_col} in the conditional coverage probabilities in \eqref{cov1} and \eqref{cov2}, then integrating over the PDF of the service distances given in Lemma~\ref{lem:joint} and Corollary~\ref{col_dist}.
\end{proof}

In the blackout case, the interference from the skipped BS (i.e., $I_1$) may be overwhelming to the SINR. Hence, interference cancellation techniques could be employed to improve the coverage probability. In this case, the interfering signal from the skipped BS is detected, demodulated, decoded and then subtracted from the received signal \cite{16a}. In this case, the coverage probability for blackout user is given by the following theorem.
\begin{theorem}
\label{theorem:Coverage Probability_IC}
Considering a PPP cellular network with BS intensity $\lambda$ in a Rayleigh fading environment, the coverage probability for blackout users with interference cancellation capabilities can be expressed as
\begin{eqnarray}
\mathcal{C}^{(IC)}_{bk}=\frac{1}{(1+\vartheta(T\text{,}\eta))^2}
\label{eq:Coverage Probability for Skipping Case with IC}
\end{eqnarray}
where $\vartheta(T\text{,}\eta)$ is defined in Lemma~\ref{lem:LT}. For the case of $\eta = 4$, \eqref{eq:Coverage Probability for Skipping Case with IC} reduces to
\begin{eqnarray}
\mathcal{C}^{(IC)}_{bk} \big|_{\eta=4}&=& \frac{1}{(1+\sqrt{T}\arctan(\sqrt{T}))^2}
\end{eqnarray}
\end{theorem}
\begin{proof}
The theorem is obtained using the same methodology for obtaining (14) but with eliminating $I_1$ from \eqref{cov2}.
\end{proof}

The coverage probability plots for best connected and blackout cases with and without nearest BS interference cancellation (IC) are shown in Fig 3. It can be observed that the simulation results (which are obtained for mobile users) are in accordance with the analysis which validates our model. The figure shows the cost of skipping in terms of coverage probability degradation. Note that the users in the BS skipping scheme alternate between the blackout and best connected cases as we allow for only one BS skipping in this paper. Hence, a user in the skipping model would spent $50\%$ of the time with the blackout coverage and $50\%$ of the time in best connected coverage. The figure also shows that interference cancellation highly improves the SINR when compared to the blackout case without interference cancellation. Note that, although the expressions in Theorems 1 and 2 are obtained using stationary PPP analysis, they totally conform with the simulations done with mobile users almost exactly, which justifies the aforementioned claim in Section I.
In the next section, we derive the effect of BS skipping on the user rate.

\begin{figure}[!t]
\centering
\includegraphics[width=0.95 \linewidth]{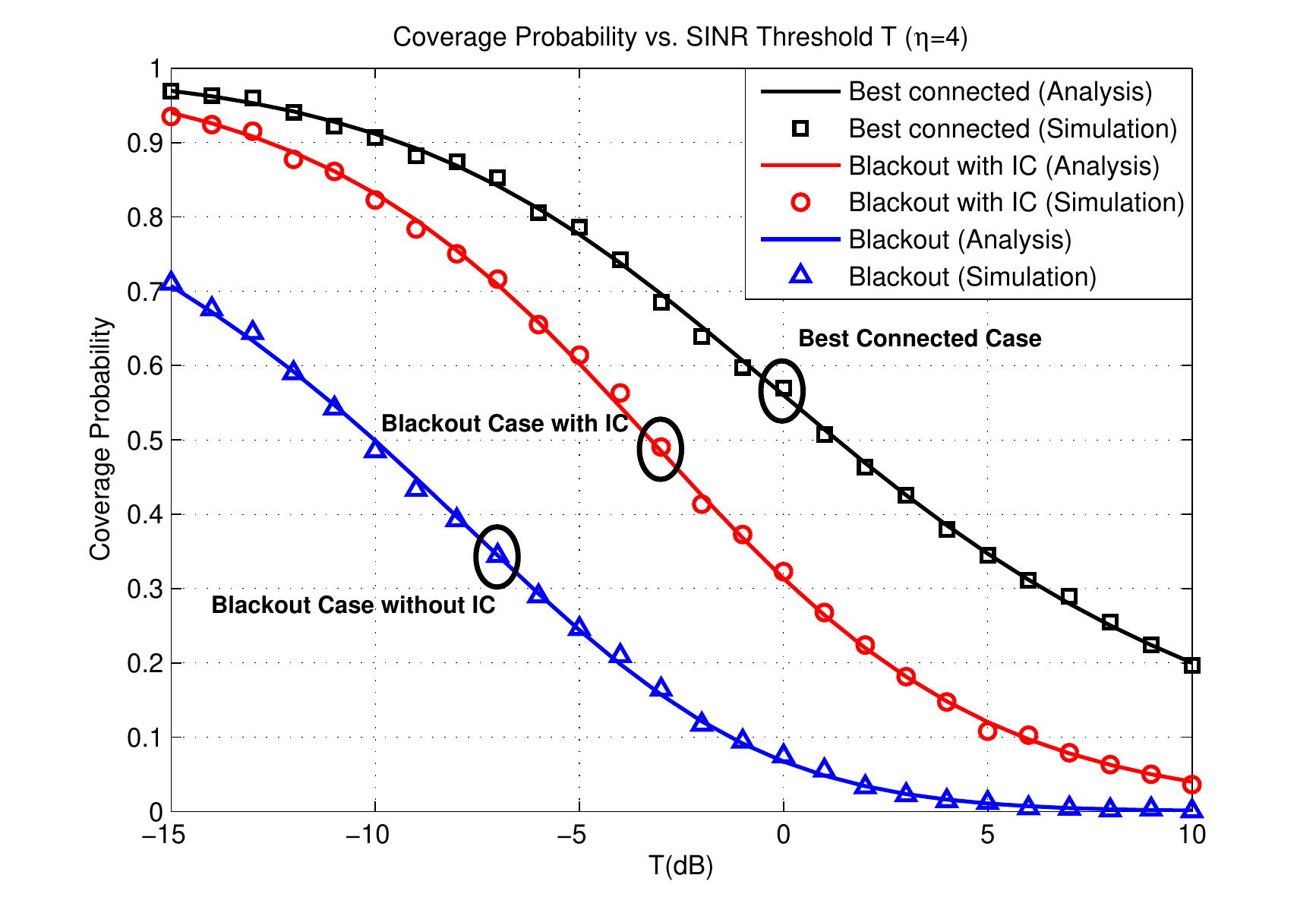}
\small \caption{Coverage probability plots for conventional and skipping cases at $\eta=4$.}
\end{figure}

\section{Handover Cost}
In this section, we encompass user mobility effect and compute handover rates for both conventional and skipping schemes. The handover rates are then used to quantify the handover delay per unit time $D_{HO}$ (i.e., time consumed in HO per time unit). $D_{HO}$ can be expressed as a function of HO rate (HOR) and HO delay $d$ as shown below
\begin{eqnarray}
D_{HO}&=& \text{HOR} * d
\end{eqnarray}
Following \cite{10a}, the HO rate in a single tier network can be expressed as
\begin{eqnarray}
H(v)&=& \frac{4v}{\pi}\sqrt{\lambda}
\end{eqnarray}
Consequently, the handover delay $D_{HO}$ for both conventional and skipping cases can be expressed as
\begin{eqnarray}
D_{HO}^{(c)}&=&\frac{4v}{\pi}\sqrt{\lambda}d\\
D_{HO}^{(sk)}&=&\frac{2v}{\pi}\sqrt{\lambda}d
\end{eqnarray}
where $D_{HO}^{(c)}$ and $D_{HO}^{(sk)}$ are the HO cost for conventional and skipping cases respectively. Note that the handover cost for the skipping case is half of the conventional case because the user skip half of the handovers across the trajectory. Assuming HO delay of 0.7 seconds for macro BSs and 2 seconds for IP-backhauled small cells \cite{14a}, we plot the handover cost in Fig.~\ref{ho_cost}.
\begin{figure}[t]
\centering
\includegraphics[width=1 \linewidth]{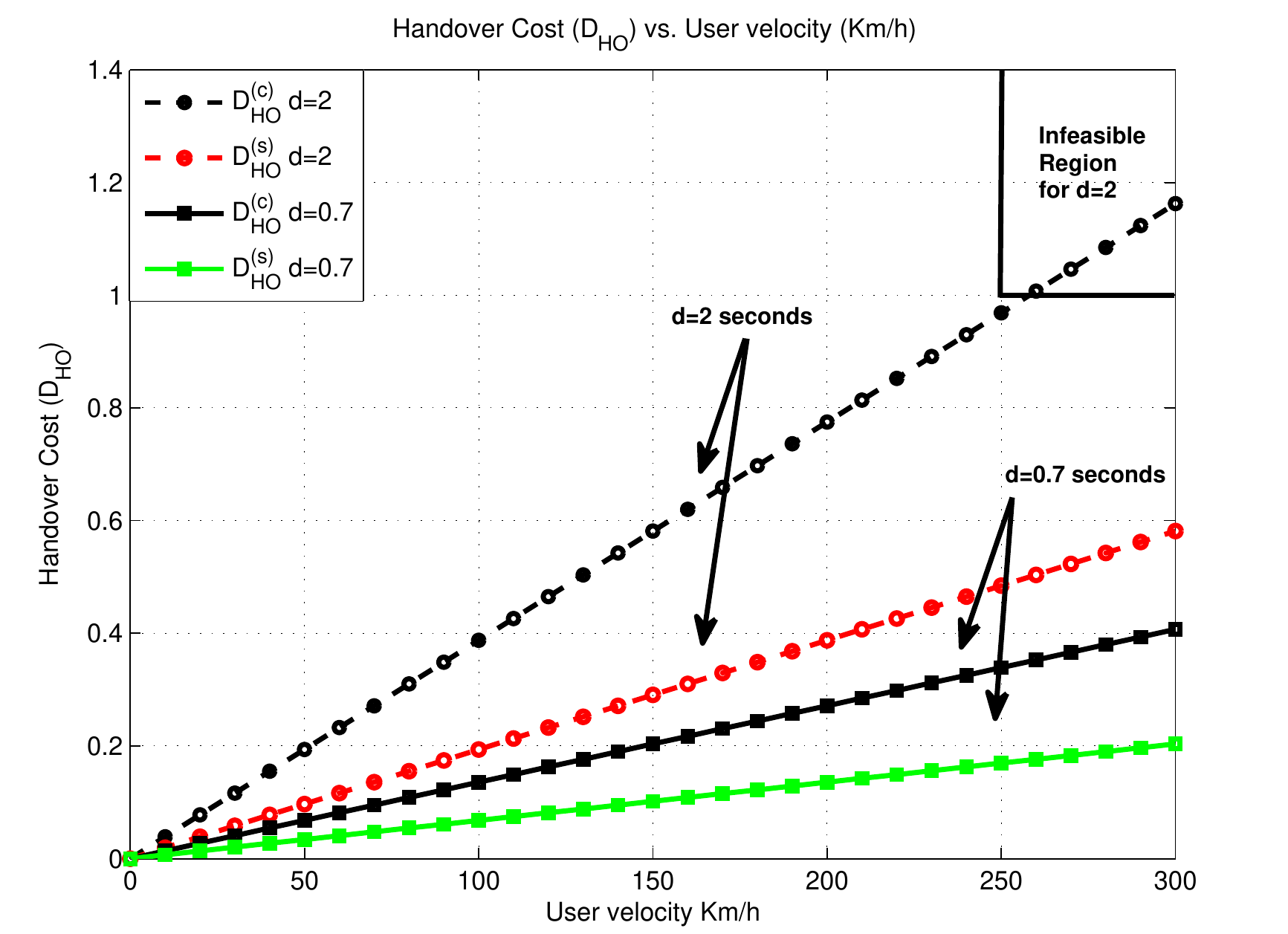}
\small \caption{$D_{HO}$ plots for conventional and HO skipping cases vs. user velocity (Kmph) for $\lambda=30 BS/km^2$.}
\label{ho_cost}
\end{figure}

\section{Average Throughput}
In this section, we derive an expression for average throughput for HO skipping case. In order to calculate the throughput, we need to omit control overhead. We assume that the control overhead consumes a fraction $u_c$ of overall network capacity which is 0.3 as per 3GPP Release 11 \cite{17a}. Thus, the average throughput (AT) can be expressed as
\begin{eqnarray}
\text{AT} &=& W\mathcal{R}(1-u_c)(1-D_{HO})
\end{eqnarray}
where $W$ is the overall bandwidth of the channel and $\mathcal{R}$ is the average spectral efficiency (i.e., nats/sec/Hz). Following \cite{6a}, the average spectral efficiency can expressed in terms of the coverage probability as
\begin{eqnarray} \label{rates}
\mathcal{R}&\stackrel{(a)}{=}& \int_{0}^{\infty}\mathbb{P}\left\{\ln(1+SINR)>z\right\}dz\\
&\stackrel{(b)}{=}&\int_{0}^{\infty}\frac{\mathbb{P}\left\{SINR>t\right\}}{t+1}dt
\end{eqnarray}
where (a) follows because $\ln(1+SINR)$ is a positive random variable and (b) follows by the change of variables $t=e^{z}-1$ \cite{18a}. For brevity, we do not show expressions for $\mathcal{R}$ for general $\eta$ as they can be directly obtained by substituting the coverage probability from Theorem 1 and Theorem 2 in \eqref{rates}. In the special case of $\eta=4$, the average spectral efficiency for the conventional and blackout cases are given by
\begin{eqnarray}
\mathcal{R}_c&=&\int_{0}^{\infty}\frac{1}{(1+t)(1-\sqrt{t}\arctan(\sqrt{t}))}dt
\end{eqnarray}
\begin{eqnarray}
\simeq 1.49\hspace{0.2cm}\text{nats/sec/Hz} \notag
\end{eqnarray}
and
\begin{eqnarray}
\mathcal{R}_{bk} &=& \int_{0}^{\infty}\frac{1-\sqrt{t}\arctan(\frac{1}{\sqrt{t}})}{(1+t)(1+\vartheta(t\text{,}4))^2}dt\\ \notag\\ \notag
&\stackrel{(c)}{\simeq}& 0.21\hspace{0.2cm} \text{nats/sec/Hz} \notag\\
&\stackrel{(d)}{\simeq}& 0.66\hspace{0.2cm} \text{nats/sec/Hz} \notag
\end{eqnarray}
where (c) and (d) are for blackout cases without and with interference cancellation, respectively.

In HO skipping case, the user alternate between the best connection and blackout case along its trajectory. More particularly, the user in HO skipping spent $50 \%$ of the time as blackout user and $50 \%$ of the time as best connected user. Hence, the average spectral efficiency for users in HO skipping is given by
\begin{eqnarray}
\mathcal{R}_{s}&=& \frac{\mathcal{R}_{c}+\mathcal{R}_{bk}}{2}\\
&\simeq&1.07\hspace{0.2cm}\text{nats/sec/Hz} \notag
\end{eqnarray}

\section{Numerical Results}
In this section, we use the developed analytical model to evaluate the performance of HO skipping in terms of users throughput. We use the following parameters to conduct our analysis. Transmission powers of all BSs are considered to be unity. Channel bandwidth ($W$) is considered to be $10 MHz$. Control overhead is assumed $0.3$ for conventional case and $0.15$ for skipping case. Analysis is conducted for $d=$ 0.7 \& 2 seconds. Path loss exponent $\eta=4$ is considered. Different values of $\lambda$ are considered here to mark the nominal speed at which HO skipping is effective. Moreover, nearest BS interference cancellation technique is considered.
\begin{figure}[!t]
\centering
\includegraphics[width=1\linewidth]{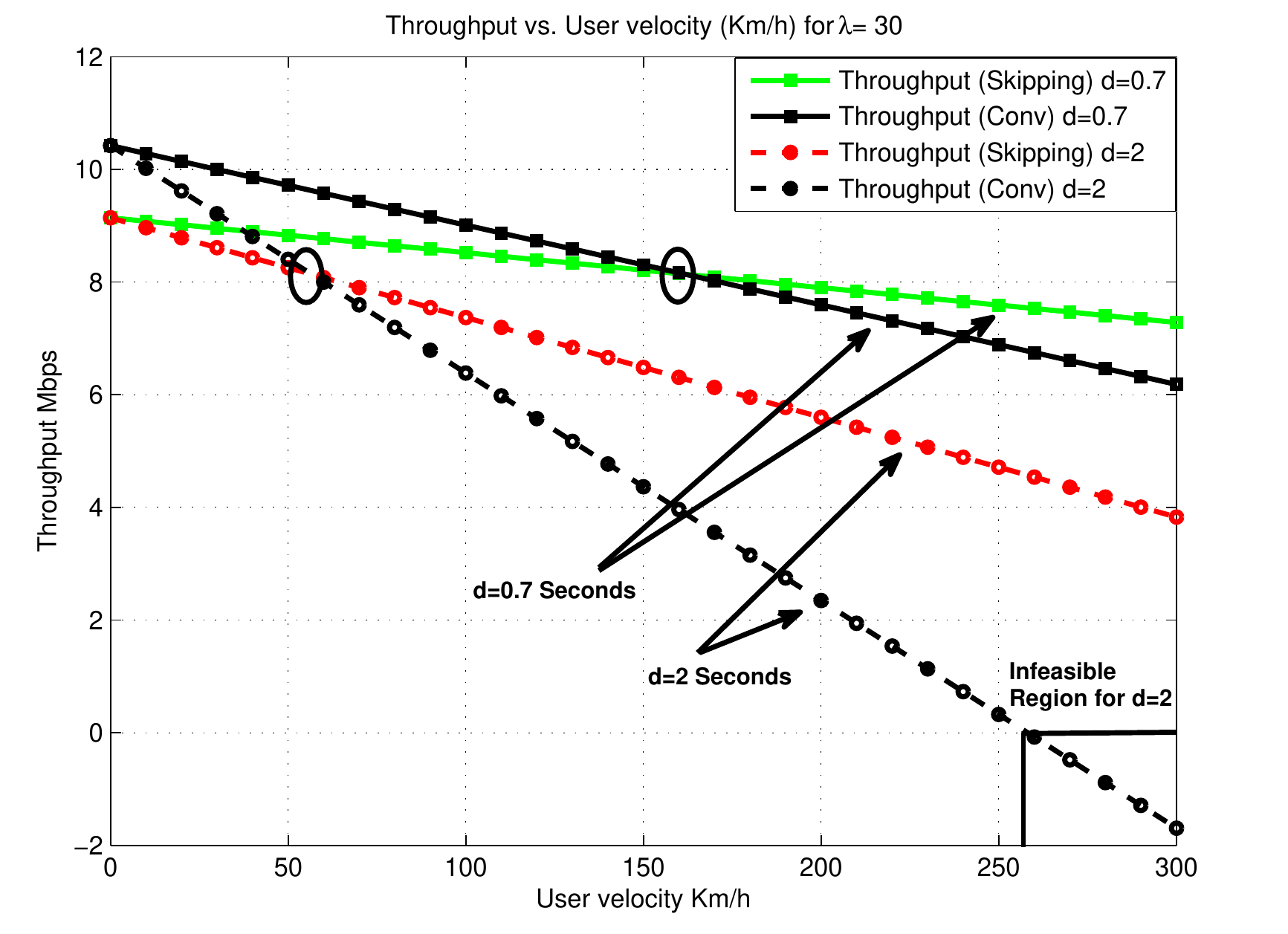}
\small \caption{Application Throughput (Mbps) vs. user velocity (Kmph) for $\lambda=30$ and $W=10MHz$.}
\end{figure}
\begin{figure}[!t]
\centering
\includegraphics[width=1 \linewidth]{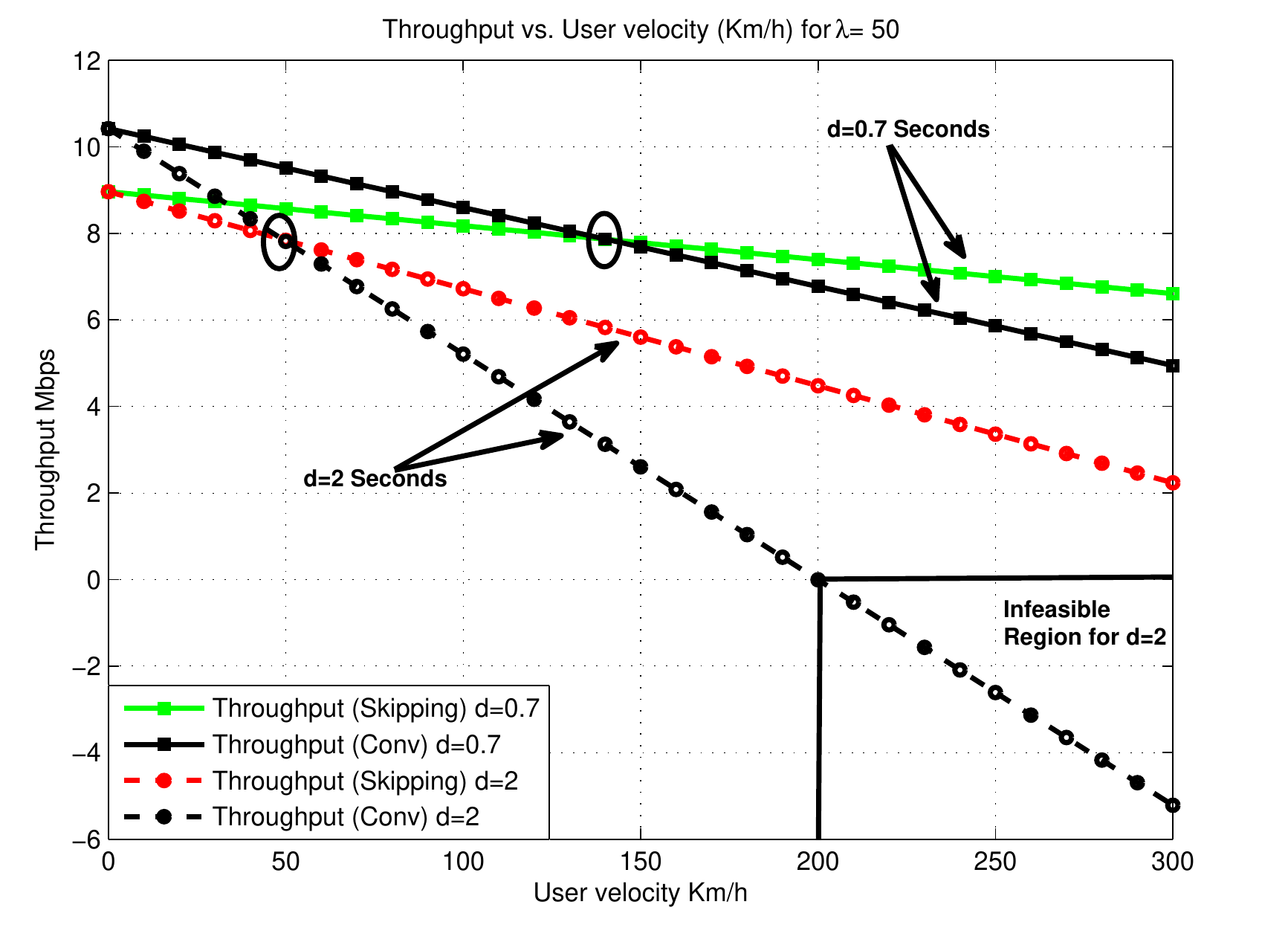}
\small \caption{Application Throughput (Mbps) vs. user velocity (Kmph) for $\lambda=50$ and $W=10MHz$.}
\end{figure}
Performance gain with HO skipping depends on BS intensity as depicted in Fig. 5-7. Note that the improved performance of the HO skipping stems from the increased HO delay with the BS intensity. From the statistics shown in Fig. 5-7, for HO skipping case, it can be observed that the user throughput experience tends to improve with the increase in user velocity. For instance, when $\lambda= 70 BS/Km^2$, the skipping HO scheme outperforms the conventional HO scheme once the user speed exceeds 40 Kmph and 110 Kmph for $d= 2$ and $0.7$ seconds respectively.

\section{Conclusion}
This paper presents a study for the effect of HO delays on the user average rate in single-tier dense cellular networks using tools from stochastic geometry. A single HO skipping scheme is proposed to reduce the negative effect of this delay. Tractable mathematical expressions for the coverage probabilities and average throughput for both the conventional HO and HO skipping scenarios are derived, which reduce to closed-forms in some special cases. Our mathematical paradigm and numerical results demonstrate the merits of the proposed HO skipping strategy in many practical scenarios. This new HO strategy enhances the Quality of Service (QOS) for service providers and helps them to offer a better service experience to both voice and data subscribers.

For future work, we will extend the study and suggest HO skipping solutions to multi-tier cellular networks. We will also evaluate the performance of multiple HO skipping and optimize number of skipped BSs to maximize the average rate performance.
\begin{figure}[!t]
\centering
\includegraphics[width=1 \linewidth]{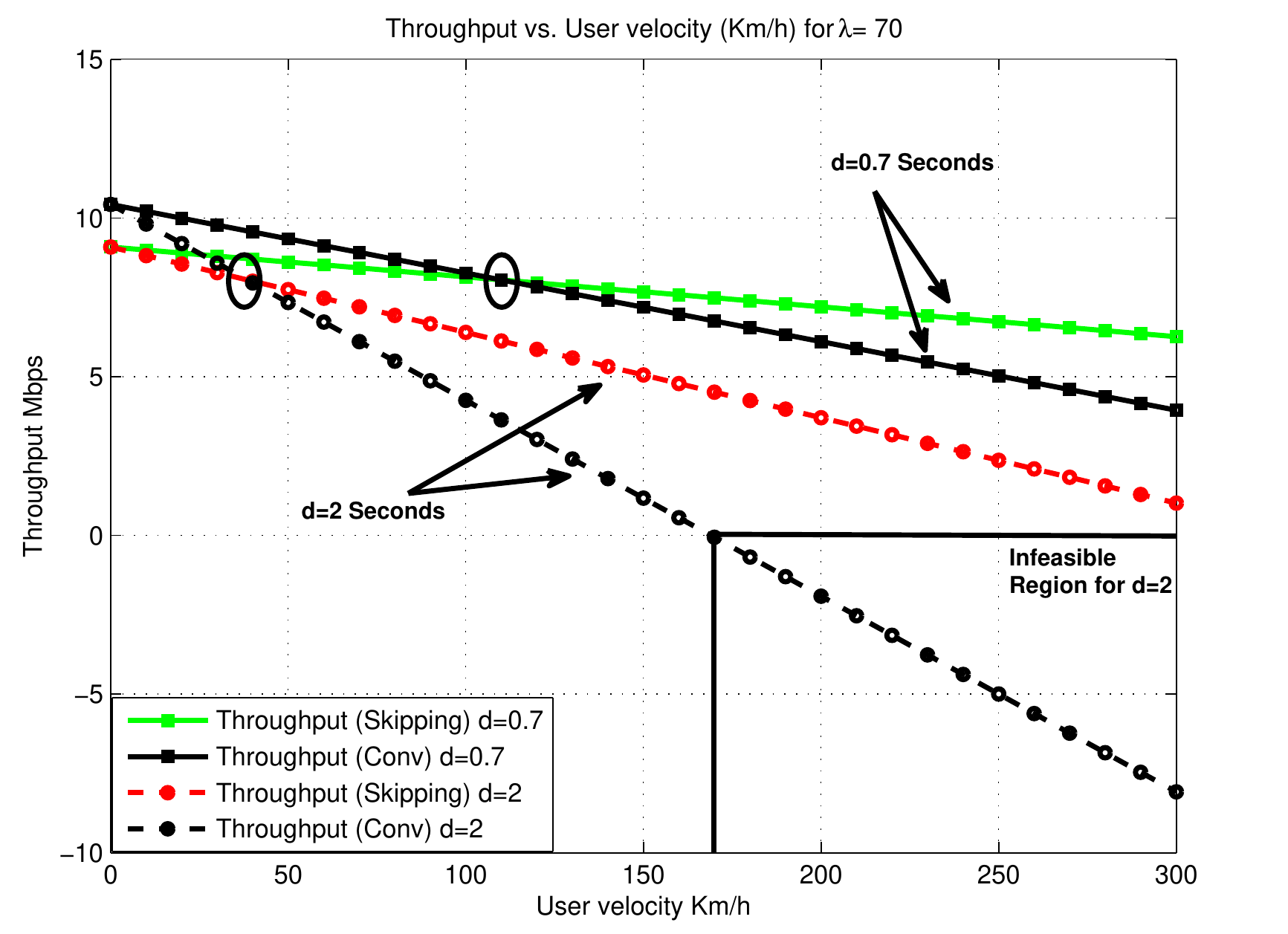}
\small \caption{Application Throughput (Mbps) vs. user velocity (Kmph) for $\lambda=70$ and $W=10MHz$.}
\end{figure}

\appendices
\ignore{
\section{Proof of Lemma 1}
The Laplace transform of $I_1$ can be expressed as
\begin{eqnarray}
\mathscr{L}_{I_1}(s) &=& \mathrm{E}\lbrack e^{-sI_1}\rbrack\notag\\
&=& \mathrm{E}\lbrack e^{-shr_{1}^{-\eta}}\rbrack\notag\\
&=& \mathrm{E}_{r_1}\lbrack \mathrm{E}_h \lbrack e^{-shr_{1}^{-\eta}}\rbrack\rbrack\notag \\
&=& \mathrm{E}_{r_1}\lbrack \mathscr{L}_h(sr_{1}^{-\eta})\rbrack\notag\\
\text{since h $\sim$ exp(1)}\notag\\
\mathscr{L}_{I_1}(s) &=& \mathrm{E}_{r_1}\lbrack\frac{1}{1+sr^{-\eta}}\rbrack\notag\\
&=& \int_{0}^{r_0}\frac{1}{1+sr^{-\eta}} f_{r_1}dr\notag\\
&=& \int_{0}^{r_0}\frac{1}{1+Tr_{0}^{\eta }r^{-\eta}}\frac{2r}{r_{0}^{2}} dr
\end{eqnarray}
}

\section{Proof of Lemma 2}
The Laplace transform of $I_r$ is given by
\begin{eqnarray}
\mathscr{L}_{I_r}(s) &=& \mathrm{E}\big\{ e^{-sI_r}\big\}\notag\\
&=&\mathrm{E}\bigg\{ e^{-s\sum_{i\epsilon \phi \backslash b_0 \text{,}b_1}{}P_{i}h_iR_{i}^{-\eta}}\bigg\}\notag
\end{eqnarray}
Due to the independence between fading and BSs location, we get:
\begin{eqnarray}
\mathscr{L}_{I_r}(s)&=&\mathrm{E}_{\phi} \bigg\{ \prod_{i\epsilon \phi \backslash b_0 \text{,}b_1}^{} \mathrm{E}_{h_i}\big\{e^{-sP_{i}h_{i} R_{i}^{-\eta}}\big\}\bigg\}\notag\\
&=&\mathrm{E}_{\phi}\bigg\{\prod_{i\epsilon \phi \backslash b_0 \text{,}b_1}^{} \mathscr{L}_{h_i}(sP_{i}R_{i}^{-\eta})\bigg\}\notag
\end{eqnarray}
But since $h_{i}$ $\sim$ exp(1), we can write
\begin{eqnarray}
\mathscr{L}_{I_r}(s)&=&\mathrm{E}_{\phi}\bigg\{\prod_{i\epsilon \phi \backslash b_0 \text{,}b_1}^{}\frac{1}{1+sP_{i}R_{i}^{-\eta}}\bigg\}\notag
\end{eqnarray}
Assuming $P_i=P$, $\forall i$ and using probability generating functional (PGFL) for PPP \cite{20a}, we get:
\begin{eqnarray}
\mathscr{L}_{I_r}(s)&=& \exp\bigg(-2\pi\lambda\int_{r_0}^{\infty}(1-\frac{1}{1+sPv^{-\eta}})vdv\bigg)\notag\\
&=& \exp\bigg(-2\pi\lambda\int_{r_0}^{\infty}\frac{1}{\frac{1}{Tr_{0}^{\eta}v^{-\eta}}+1}vdv\bigg)\notag\\ \notag
\end{eqnarray}
By change of variables $w=\bigg(\frac{v}{T^{1/ \eta}r_0}\bigg)^2$ and substituting $s=\frac{Tr_{0}^{\eta}}{P}$, we get
\begin{eqnarray}
\mathscr{L}_{I_r}\left(\frac{Tr_{0}^{\eta}}{P}\right)
&=& \exp\left(-\pi \lambda r_{0}^{2} T^{2/\eta}\int_{T^{-2/\eta}}^{\infty}\frac{1}{1+w^{\eta/2}}dw\right)\notag\\ \notag\\
&=& \exp\left(- \pi\lambda r_{0}^{2}\vartheta(T\text{,}\eta)\right)
\end{eqnarray}
where
\begin{eqnarray}
\vartheta(T\text{,}\eta) &=& T^{2/\eta}\int_{T^{-2/\eta}}^{\infty}\frac{1}{1+w^{\eta/2}}dw\notag
\end{eqnarray}
Laplace transform of $I_1$ can be expressed as
\begin{eqnarray}
\mathscr{L}_{I_1}(s) &=& \mathrm{E}\big\{ e^{-sI_1}\big\}\notag\\
&=& \mathrm{E}\big\{ e^{-sPhr_{1}^{-\eta}}\big\}\notag\\
&=& \mathrm{E}_{r_1}\big\{ \mathrm{E}_h \lbrack e^{-sPhr_{1}^{-\eta}}\rbrack\big\}\notag \\
&=& \mathrm{E}_{r_1}\big\{ \mathscr{L}_h(sPr_{1}^{-\eta})\big\}\notag
\end{eqnarray}
But since h $\sim$ exp(1), we get:
\begin{eqnarray}
\mathscr{L}_{I_1}(s) &=& \mathrm{E}_{r_1}\big\{\frac{1}{1+sPr^{-\eta}}\big\}\notag\\
&=& \int_{0}^{r_0}\frac{1}{1+sPr^{-\eta}} f_{r_1}(r)dr\notag
\end{eqnarray}
Using conditional distribution of $r_1$ obtained in corollary 1, we get:
\begin{eqnarray}
\mathscr{L}_{I_1}\bigg(\frac{Tr_{0}^{\eta}}{P}\bigg)&=& \int_{0}^{r_0}\frac{1}{1+Tr_{0}^{\eta }r^{-\eta}}\frac{2r}{r_{0}^{2}} dr\\ \notag
\end{eqnarray}

\bibliographystyle{IEEEtran}
\bibliography{IEEEabrv,mybibr}
\vfill

\end{document}